\documentclass[amsthm]{elsart}

\usepackage{yjsco}
\usepackage{natbib}
\usepackage{amsmath}
\usepackage{hyperref}
\usepackage{amssymb}
\usepackage{amsthm}
\usepackage{enumitem}
\usepackage {tikz}
\usetikzlibrary {positioning}

\renewcommand{\theenumii}{\theenumi.\arabic{enumii}.}

\newtheorem{theorem}{Theorem}[section]
\theoremstyle{definition}
\newtheorem{definition}{Definition}[section]
\newtheorem{example}{Example}[section]
\newtheorem{corollary}{Corollary}[section]

\newtheorem{alg}{Algorithm}[section]
\theoremstyle{remark}
\newtheorem{remark}{Remark}[section]

\begin{document}

\begin{frontmatter}

\title{
Computing Hypergeometric Solutions of 
Second Order Linear Differential Equations 
using Quotients of Formal Solutions
and Integral Bases
}


\author{Erdal Imamoglu}\footnotemark[1]
\address{	Department of Mathematics, 
	Florida State University, 
	Tallahassee, FL 32306, USA.}
\ead{eimamogl@math.fsu.edu}

\author{Mark van Hoeij}\footnote{Supported by NSF grant 1319547.}
\address{	Department of Mathematics, 
	Florida State University, 
	Tallahassee, FL 32306, USA.}
\ead{hoeij@math.fsu.edu}

\begin{abstract}
We present two algorithms for computing hypergeometric solutions of second order linear differential operators with rational function coefficients. Our first algorithm searches for solutions of the form
\begin{equation}
	\label{form_of_solutions}
	\exp(\int r \, dx)\cdot{_{2}F_1}(a_1,a_2;b_1;f)
\end{equation}
where $r,f \in \overline{\mathbb{Q}(x)}$, and $a_1,a_2,b_1 \in \mathbb{Q}$. It uses modular reduction and Hensel lifting. 
Our second algorithm tries to find solutions in the form
\begin{equation}
	\label{form_of_solutions_general_case}
	\exp(\int r \, dx)\cdot 
	\left( 
		r_0 \cdot{_{2}F_1}(a_1,a_2;b_1;f)
		+ 
		r_1 \cdot{_{2}F_1}'(a_1,a_2;b_1;f)
	\right)
\end{equation}where $r_0, r_1 \in \overline{\mathbb{Q}(x)}$, as follows:
It tries to transform the input equation to another equation with solutions of type (\ref{form_of_solutions}), and then uses the first algorithm. 
\end{abstract}


\end{frontmatter}

\section{Introduction}
\label{section_introduction}
A second order homogeneous linear differential equation with rational function coefficients $A_i \in  \mathbb{Q}(x)$
\begin{equation}
	\label{differential_equation}
	A_2 y'' + A_1 y' + A_0 y = 0
\end{equation}
corresponds to the differential operator
\begin{equation*}
	L = A_2 \partial^2 + A_1 \partial + A_0 \in \mathbb{Q}(x)[\partial]
\end{equation*}
where $\partial = \frac{d}{dx}$. Another representation of (\ref{differential_equation}) is 
\begin{equation*}
	L(y) =0.
\end{equation*}

This paper gives two heuristic (see Remarks \ref{address} and \ref{address2}) algorithms to find a hypergeometric solution of (\ref{differential_equation}) in the form (\ref{form_of_solutions}) and (\ref{form_of_solutions_general_case}). The form (\ref{form_of_solutions_general_case}) is more general than in prior works \cite{fang_vhoeij, kunwar_vhoeij, vhoeij_vidunas, kunwar}. Papers \cite{vhoeij_vidunas}  and \cite{kunwar} were restricted to a specific number of singularities (4 in \cite{vhoeij_vidunas}  and 5 in \cite{kunwar}). Papers \cite{kunwar_vhoeij} and \cite{fang_vhoeij} were restricted to specific degrees (degree 3 in \cite{kunwar_vhoeij} and a degree-2 decomposition in \cite{fang_vhoeij}). Our algorithms are not restricted to a specific number of singularities or a specific degree. Moreover, our algorithms can find algebraic functions $f$ in (\ref{form_of_solutions}) and (\ref{form_of_solutions_general_case}).

Our first algorithm, Algorithm \ref{find2f1}, tries to find solutions of (\ref{differential_equation}) in the form of (\ref{form_of_solutions}). Our second algorithm, Algorithm \ref{hypergeometricsols}, tries to reduce equations with solutions of form (\ref{form_of_solutions_general_case}) to equations and then calls our first algorithm. 

We assume that (\ref{differential_equation}) has no Liouvillian solutions and hence irreducible. Otherwise one can solve (\ref{differential_equation}) with Kovacic's algorithm \cite{kovacic}.

Let $L_{inp}$ $\in$ $\mathbb{Q}(x)[\partial]$ ($L$ input) be a second order linear differential operator, regular singular (details in Section \ref{subsection_differential_operators}) and without Liouvillian solutions:
\begin{itemize}
	\item Goal 1: (Algorithm \ref{find2f1}) Find a solution of form~(\ref{form_of_solutions}) if it exists.
	\item  Goal 2: (Algorithm \ref{hypergeometricsols}) Try to transform $L_{inp}$ to a simpler operator (which hopefully has a solution in form (\ref{form_of_solutions})).
\end{itemize}

The crucial steps for Goal 1 are to find (candidates for) $a_1$, $a_2$, $b_1$ and the \emph{pullback function} $f$. Finding the parameters $a_1$, $a_2$, $b_1$ is the combinatorial part; Theorem \ref{rh_for_des_thm} helps us to eliminate the vast majority of cases. Given $a_1,a_2,b_1$ (or equivalently a base operator $L_B$), \emph{if we know the value of a certain constant $c$}, by comparing \emph{quotients of formal solutions} of $L_B$ and $L_{inp}$, we can compute $f$. We have no direct formula for $c$; to obtain it with a finite computation, we take a prime number $\ell$. Then, for each $c \in \{1, \dots , \ell-1 \}$ we try to compute $f$ modulo $\ell$. If this succeeds, then we lift $f$ modulo a power of $\ell$, and try rational number reconstruction.

Goal 2 is to find a transformation to convert $L_{inp}$ to a simpler operator $\tilde{L}_{inp}$. The key idea is to follow the strategy of the \texttt{POLRED} algorithm in \cite{cohen}. It takes as input a polynomial $L_{inp} \in \mathbb{Q}[x]$ and finds an element $\mathcal{G} \in \mathbb{Q}[x] / \mathbb{Q}[x]L_{inp}$ whose minimal polynomial is close to optimal. It works as follows:
\begin{enumerate}
	\item ``Finite points'': Compute an integral basis.
	\item ``Valuations at infinity'': Find an integral element $\mathcal{G}$ with (near) optimal absolute values.
\end{enumerate}
Our key idea is to apply \texttt{POLRED}'s strategy to $L_{inp} \in \mathbb{Q}(x)[\partial]$. First compute an integral basis (introduced in \cite{kauers_koutschan}). Then \emph{normalize at infinity} (following \cite[Section 2.3]{trager} where this is done for function fields) so we can select an element with minimal valuations at infinity.

\begin{example}[Rational Pullback Function]\label{example1}
The differential operator
\begin{equation*}
	L_{inp} = 147x(x-1)(x+1)\partial^2+ (266x^2 - 42x - 98)\partial + 20x-5
\end{equation*}
has a $_{2}F_1$-type solution in the form of (\ref{form_of_solutions}), which is
\begin{equation*}
	Y(x) = 
	\exp{(\int r \, dx)} \cdot 
	{_2F_1} \left( \frac{5}{42}, \frac{11}{42} ; \frac{2}{3} ; f \right)
\end{equation*}
where 
\begin{equation}
	\label{expr}
	\exp({\int r \, dx}) = \left( x+1 \right) ^{-{\frac {5}{21}}}
	\,\,\,\,\,{\rm and }\,\,\,\,\,
	f = {\frac {4x}{ \left( x+1 \right) ^{2}}}
\end{equation}
Section \ref{subsection_cand_exp_diffs} shows how to find the parameters $a_1, a_2, b_1$ $=$ $\frac{5}{42}, \frac{11}{42}, \frac{2}{3}$. Then $f$ is computed with the quotient method:
\end{example}

\begin{remark}[The Quotient Method]\label{qmethod}
The hypergeometric function 
\begin{equation*}
	_2F_1 \left( \frac{5}{42}, \frac{11}{42}; \frac{2}{3} ; x \right)
\end{equation*}
is a solution of the Gauss hypergeometric differential operator
\begin{equation*}
	L_B = 
	{{\partial}}^{2}
	+{\frac { \left( 29\,x-14 \right) }{21x\left( x-1 \right) }} \partial 
	+{\frac {55}{1764\,x \left( x-1 \right) }}.
\end{equation*}
$L_B$ has two solutions at $x=0$. They are
\begin{align*}
	y_1(x) & ={_2F_1}\left( \frac{5}{41}, \frac{11}{42} ; \frac{2}{3}, x \right) = 1+{\frac {55}{1176}}x+ \dots , \\ 
	y_2(x) & = {x^{\frac{1}{3}}} \left(1 + {\frac {475\,{}}{2352}}x+{\frac {1941325\,}{19361664}}x^2 + \dots \right). 
\end{align*}
The so-called {\em exponents} of $L_B$ at $x=0$ are the exponents of $x$ in the dominant terms of the solutions $y_1$ and $y_2$,  so the exponents are $e_{0,1}=0$ and $e_{0,2}=\frac{1}{3}$. The minimal operator for $y(f)$ has the following solutions at $x=0$:
\begin{align*}
	y_1(f) & =  1+{\frac {55}{294}}x-{\frac {4939}{86436}}{x}^{2}+{\frac {16135823}{304946208}}{x}^{3}+  \dots \,,\\
	y_2(f) & =  c \cdot {x^{\frac{1}{3}}} \left( 1+{\frac {83{}}{588}}x+{\frac {6805}{1210104}}x^2+ \dots \right)
\end{align*}
for some constant $c$ that depends on $f$. The exponents are again $0$ and $\frac{1}{3}$, because $x=0$ is a root of $f$ with multiplicity $e_0=1$ (Theorem \ref{theorem_expdiffs_and_multiplicity}). Let
\begin{align}
	\label{solinp1}
	Y_1(x) & = \exp( \int r \,dx ) y_1(f) = 1-{\frac {5}{98}}x+{\frac {439}{9604}}{x}^{2}+ \dots \, ,\\
	\label{solinp2}
	Y_2(x) & = \exp( \int r \,dx )y_2(f) = c \cdot x^{\frac{1}{3}} \left( 1 - {\frac {19}{196}}x +\dots \right).
\end{align}
(\ref{solinp1}) and (\ref{solinp2}) form a basis of solutions of $L_{inp}$. Here $\exp( \int r \,dx )$ is the same as in (\ref{expr}). Denote the quotients of the formal solutions of $L_B$ and $L_{inp}$ by
\begin{align*}
	q & = \frac{y_1(x)}{y_2(x)} \\ 
	Q & = \frac{Y_1(x)}{Y_2(x)}  = \frac{y_1(f)}{y_2(f)}=q(f)
\end{align*}
respectively. It follows that $q^{-1}(Q(x))$ gives a series expansion of $f$ at $x=0$. Given enough terms we can compute $f$ with rational function reconstruction. This {\em Quotient Method} was already used in \cite[Section 5.1]{vhoeij_vidunas}. In order to turn this into an algorithm for solving differential equations we need to answer the following questions:
\begin{enumerate}[label=Q\arabic*.]
	\item How many terms are needed to reconstruct $f$? 
	This is equivalent to finding a degree bound for $f$.
	\item How to find the parameters $a_1$, $a_2$, $b_1$? 
	This is the combinatorial part of our algorithm.
	\item The exponents $0$, $\frac{1}{3}$ of $L_{inp}$ at $x=0$ only determine $\frac{Y_1}{Y_2}$ up to a constant factor (see Remark \ref{non-rem-sing} in Section \ref{subsection_transformations_and_singularities}). This means the quotient $\frac{y_1(f)}{y_2(f)}$ is only known up to a constant $c$. How to find this constant?	
	\item What if $L_{inp}$ has logarithmic solutions at $x=0$?
	\item What if $f$ is an algebraic function?
	\item What if $L_{inp}$ does not have solutions in the form of (\ref{form_of_solutions}), but has solutions in the form of (\ref{form_of_solutions_general_case})?
\end{enumerate}

Remark \ref{ansQ1}, Section \ref{subsection_cand_exp_diffs}, Section \ref{subsection_recovering_f_and_r}, and Section \ref{subsection_logcase_quotient_method} provide answers to Q1, Q2, Q3, and Q4 respectively. Section \ref{subsection_recovering_f_and_r} answers Q5.  The parts Q1,...,Q5 are already in our ISSAC 2015 paper \cite{imamoglu_vhoeij}. Algorithm \ref{hypergeometricsols}, which finds solutions of form (\ref{form_of_solutions_general_case}), is new compared to \cite{imamoglu_vhoeij}. So Q6 is the main new part in this paper. It will be discussed in Section \ref{section_integral_basis}. Example \ref{Q6example} will illustrate to Q6.
\end{remark}

\begin{remark}\label{address}
	Both algorithms are very effective in practice but they are not proven. 
	For completeness for Goal 1 we still need a 
	theorem for \emph{good prime numbers}. A
	\emph{good prime} is a prime for which reconstruction will work.
\end{remark}

\begin{example}[Algebraic Pullback Function]\label{example2}
The differential operator
\begin{equation*}
	L_{inp} = 
	{{\partial}}^{2}
	+\frac{1}{4}\,{\frac {{x}^{4}-44\,{x}^{3}+1206\,{x}^{2}-44\,x+1}{ \left( {x}^{2}-34\,x+1 \right)^{2}{x}^{2}}}
\end{equation*}
has a $_2F_1$-type solution in the form of (\ref{form_of_solutions}), which is
\begin{equation*}
	Y(x) = \exp( -\frac{1}{2} \int r \,dx )\, \cdot \, _2F_1 \left( \frac{1}{3}, \frac{2}{3} ; 1 ; f \right)
\end{equation*}
where $r =$
\begin{equation*}
	{\frac {-{x}^{5}+22\,{x}^{4}-55\,{x}^{3}-343\,{x}^{2}+ 58\,x-1+6\,x\left( {x}^{2}-7\,x+1 \right)\sqrt {{x}^{2}-34\,x+1 }}{x \left({x}^{4}-41\,{x}^{3}+240\,{x}^{2}-41\,x+1 \right) \left( x+1 \right) }}
\end{equation*}
and
\begin{equation*}
	f =\frac{1}{2}\,{\frac {1+30\,x-24\,{x}^{2}+{x}^{3}-\left( {x}^{2}-7\,x+1 \right) \sqrt {{x}^{2}-34\,x+1 }}{1+3\,x+3\,{x}^{2}+{x}^{3}}}.
\end{equation*}
Here the pullback function $f$ is an algebraic function: $\mathbb{Q}(x,f)$ is an algebraic extension of $\mathbb{Q}(x)$ of degree $a_f = 2$ ($a_f$ is $1$ if and only if $f$ is a rational function, as in Example \ref{example1}).
Algorithm \ref{find2f1} can find this solution.
\end{example}

\begin{example}[Finding Solutions in the form of (\ref{form_of_solutions_general_case}) using an Integral Basis]\label{Q6example}
Consider the differential operator\footnote{Prof. Jean-Marie Maillard sent us this differential operator.} 
\begin{align*}
L_{inp} = {{\partial}}^{2}
&-{\frac { 512\,{x}^{5}+384\,{x}^{4}-64\,{x}^{3}-88\,{x}^{2}-10\,x-1  }{x \left( 4\,x-1 \right)  \left( 4\,x+1 \right)  \left( 16\,{x}^{3}+24\,{x}^{2}+5\,x+1 \right) }}{\partial}\\
&+{\frac {512\,{x}^{5}+64\,{x}^{4}-128\,{x}^{3}-60\,{x}^{2}-8\,x-1}{{x}^{2} \left( 4\,x-1 \right)  \left( 4\,x+1 \right)  \left( 16\,{x}^{3}+24\,{x}^{2}+5\,x+1 \right) }}.
\end{align*} 

Algorithm \ref{find2f1} can not solve $L_{inp}$. We try to transform $L_{inp}$ to simpler operator $\tilde{L}_{inp}$.
First we compute an integral basis. Then we normalize the basis at infinity and obtain $[B_0, B_1]$ where
\begin{align*}
	B_0 = 
	&{\frac { 16\,{x}^{4}-{x}^{2} }{ \left( 16\,{x}^{3}+24\,{x}^{2}+5\,x+1 \right) x}}{\partial}\\
	&+{\frac {-34359738400\,{x}^{3}-51539607556\,{x}^{2}-10737418241\,x-2147483648}{ \left( 16\,{x}^{3}+24\,{x}^{2}+5\,x+1 \right) x}}
\end{align*}
and
\begin{align*}
	B_1 = {\frac {  16\,{x}^{3}-x }{ \left( 16\,{x}^{3}+24\,{x}^{2}+5\,x+1 \right) x}}{\partial}+{\frac {-32\,{x}^{2}-4\,x-1}{ \left( 16\,{x}^{3}+24\,{x}^{2}+5\,x+1 \right) x}}.
\end{align*}
We try to find a suitable $\mathcal{G} \in \mathbb{Q}(x)[\partial] / \mathbb{Q}(x)[\partial]L_{inp}$. It should be a combination of $B_0$ and $B_1$. For this example, we take $\mathcal{G} = B_1$. This $\mathcal{G}$ is called a \emph{gauge} transformation. It maps solutions of $L_{inp}$ to solutions of
\begin{align*}
	\tilde{L}_{inp} = {{\partial}}^{2}
	&+\frac{48\,x^2 - 1}{x\, (16\,x^2-1)}\partial+ \frac{16}{16\,{x}^{2}-1}.
\end{align*}
$\tilde{L}_{inp}$ has a solution in the form of (\ref{form_of_solutions}),
\begin{equation*}
y(x) = {_2F_1}\left( \frac{1}{2}, \frac{1}{2} ; 1 ; 16x^2 \right)
\end{equation*}
which is easy to find with Algorithm \ref{find2f1}. Then we apply the inverse \emph{gauge} transformation and obtain a solution of $L_{inp}$ in the form of (\ref{form_of_solutions_general_case}), which is
\begin{align*}
	Y(x) = \left( {4{x}^{3}+x^2+\frac{x}{2}} \right) {_2F_1}\left( \frac{1}{2}, \frac{1}{2} ; 1 ; 16x^2 \right)+ \left( 32{x}^{5}-2{x}^{3}  \right) {_2F_1}\left( \frac{3}{2}, \frac{3}{2} ; 2 ; 16x^2 \right).
\end{align*}
	
\end{example}

\section{Preliminaries}
\label{section_prelimineries}
This section recalls the concepts needed in later sections.

\subsection{Differential Operators, Singularities, Formal Solutions}
\label{subsection_differential_operators}
We start with some classical definitions which can be also found in \cite{ince, vanderput_singer, kunwar, fang, debeerst, yuan}.
\begin{definition}
Let $L = \sum_{i=0}^{n} A_i \partial^i \in  \mathbb{C}(x)[\partial]$ be an operator of order $n$.
\begin{enumerate}[label=(\roman*)]
	\item A point $p \in \mathbb{C}$ is called a \emph{singularity} of $L$ if it is a zero of the leading coefficient of $L$ or a pole of any other coefficient of $L$. The point $p={\infty}$ is called a singularity if $p=0$ is a singularity of $L_{1/x}$. Here $L_{1/x}$ is the differential operator obtained from $L$ via a change of variables $x \mapsto \frac{1}{x}$ (note that $x \mapsto f$ sends $\partial$ to $\frac1{f'} \partial)$. 
	\item If $x=p$ is not a singularity, then it is called a \emph{regular point} of $L$.
	\item A singularity $p \in \mathbb{C}$ is called a \emph{regular singularity} if $(x-p)^i \, \frac{A_{n-i}}{A_n}$ is analytic at $x=p$ for $1 \leq i \leq n-1$. The point $p={\infty}$ is  a regular singularity if $p=0$ is a regular singularity of $L_{1/x}$. 
	\item $L$ is \emph{regular singular} if all its singularities are regular singular.
\end{enumerate} 
\end{definition}

\subsection{Gauss Hypergeometric Differential Operator and $_2F_1$ Function}\label{subsection_gauss_hypergeometric_function}
Definitions also can be found in \cite{wang_guo, kunwar, fang}. Let $a_1$, $a_2$, $b_1$ $\in$ $\mathbb{Q}$. The operator
\begin{equation*}
	L_B = x(1-x)\partial^2 + (b_1-(a_1+a_2+1)x)\partial - a_1a_2
\end{equation*}
is called \emph{Gauss hypergeometric differential operator} (GHDO). The solution space of $L_B$ in a \emph{universal extension} \cite{fang} has dimension 2 because the order of $L_B$ is 2. One of the solutions of $L_B$ at $x=0$ is the \emph{Gauss hypergeometric function}. It is denoted by $_{2}F_1$ and defined by the Gauss hypergeometric series
\begin{equation*}
	_{2}F_1 (a_1,a_2;b_1;x) 
	= 
	\sum _{k=0}^{\infty} \frac{(a_1)_k(a_2)_k}{(b_1)_k k! }x^k.
\end{equation*}
Here $(\lambda)_k$ denotes the \emph{Pochammer symbol}. It is defined as $(\lambda)_k=\lambda(\lambda+1)\dots(\lambda+k-1)$ and $(\lambda)_0 = 1$. 

$L_B$ has three regular singularities at the points $x=0$, $x=1$, and $x=\infty$, with exponents $\{0,1-b_1\}$, $\{0, b_1-a_1-a_2\}$, and $\{a_1,a_2\}$ respectively. We denote the exponent differences of a GHDO as $\alpha_0 = |1-b_1|$, $\alpha_1 = |b_1-a_1-a_2|$, $\alpha_{\infty} = |a_1-a_2|$. We may assume $\{ a_1, a_2, b_1-a_1, b_1 - a_2 \} \cap \mathbb{Z} = \emptyset$, otherwise $L_B$ is reducible (it has exponential solutions).

Let $d_i$ be $\infty$ if $\alpha_i \in \mathbb{Z}$, and the denominator of $\alpha_i$ if $\alpha_i \in \mathbb{Q} - \mathbb{Z}$. We will only consider $a_1,a_2,b_1$ for which $L_B$ has no Liouvillian solutions. From the Schwarz list \cite{schwarz} one finds that this is equivalent to $\frac{1}{d_0} + \frac{1}{d_1} + \frac{1}{d_{\infty}} < 1$.

\subsection{Transformations and Singularities}\label{subsection_transformations_and_singularities}
We summarize properties of transformations in this section. These properties can also be found in \cite{kunwar, fang, debeerst, yuan}.

Let $L_1 \in \mathbb{C}(x)[\partial]$ 
be a differential operator of order $2$, and 
let $y$ be a solution of $L_1$.
We consider the following {\em transformations} 
that send solutions of $L_1$ to solutions of another second order differential operator $L_2$.
\begin{enumerate}
	\item \emph{Change of variables}: 
	\\$y(x) \longrightarrow y(f)$, 
	where {$f\in \overline{\mathbb{Q}(x)}$}.  
	\\For $L_1$ this means substituting $(x, \partial) \mapsto (f, \frac1{f'} \partial)$.
	\\Notation: 
	$L_1 \hspace{1mm}{\xrightarrow{f}}_C
	\hspace{1mm} L_2$.
	
	\item \emph{Exp-product}: 
	\\$y(x) \longrightarrow \exp{(\int{r \,dx})}y(x)$, 
	where {$r\in \overline{\mathbb{Q}(x)}$}. 
	\\For $L_1$ this means $\partial \mapsto \partial - r$.
	\\Notation: 
	$L_1 \hspace{1mm}{\xrightarrow{r}}_E
	\hspace{1mm} L_2$.
	
	\item \emph{Gauge transformation}:
	\\$y(x) \longrightarrow r_0 \cdot y(x) + r_1 \cdot y(x)$, 
	where {$r_0, r_1 \in \mathbb{Q}(x)$}. 
	\\For $L_1$ this means computing 
	the least common left multiple of $L_1$ and $r_1\partial + r_0$, and 
	right-dividing it by $r_1\partial + r_0$.
	\\Notation: 
	$L_1 \hspace{1mm}{\xrightarrow{r_0, r_1}}_G
	\hspace{1mm} L_2$.
\end{enumerate}

\begin{remark}
\label{defandrem}
Transformations can affect singularities and exponents.
\begin{enumerate}[label=(\roman*)]
	\item If a transformation ${\xrightarrow{r}}_E$ can send a singular point $x=p$ to a regular point $x=p$, then we call $x=p$ a {\em false singularity}. 
	\item A singularity $x=p$ is a false singularity \cite{debeerst} if and only if $x=p$ is not logarithmic and the exponent difference is 1.
	\item If $x=p$ is a singularity of $L_1$ and if transformation ${\xrightarrow{r}}_E$ can send $L_1$ to an equation $L_2$ for which all solutions of $L_2$ are analytic at $x=p$, then we call $x=p$ a {\em removable singularity}. 
	\item A point $x=p$ is removable  \cite{debeerst} if and only if $x=p$ is not logarithmic and the exponent difference is an integer. Non-removable singularities are called {\em true singularities}.
	\item A point $x=p$ is a true singularity if and only if the exponent difference is not an integer {\em or} $x=p$ is logarithmic.
	\item  If $L_1 \hspace{1mm}{\xrightarrow{r_0, r_1}}_G \hspace{1mm} L_2$, then $L_1$ and $L_2$ are called \emph{gauge equivalent}. If $V_1$ and $V_2$ are the solution spaces of $L_1$ and $L_2$ respectively, then $\mathcal{G}=r_1 \partial + r_0$ maps $V_1$ to $V_2$, i.e., $\mathcal{G}(V_1)=V_2$.
\end{enumerate}
\end{remark}

\begin{remark}\label{non-rem-sing}
The quotient method (Remark \ref{qmethod} in Section \ref{section_introduction}) can only use true singularities, otherwise, $\frac{Y_1}{Y_2}$, the quotients of solutions of $L_{inp}$, would only be known up to a M\"obius transformation instead of up to a constant.
\end{remark}

\begin{remark}
At the moment, Algorithms \ref{find2f1} and \ref{hypergeometricsols} are only implemented for rational function coefficients. However, if $f,r,r_0,r_1$ are algebraic, then the three transformations may turn an operator with rational function coefficients into an operator with algebraic function coefficients. 
\end{remark}

\begin{theorem}\label{theorem_expdiffs_and_multiplicity}\cite{bostan_chyzak_vhoeij_pech}
Let the GHDO $L_B$ have exponent differences $\alpha_0$ at $x=0$, $\alpha_1$ at $x=1$, and $\alpha_{\infty}$ at $x=\infty$. Let $L_B \hspace{1mm} {\xrightarrow{f}}_C \hspace{1mm} L_{inp}$. If $f(p) \in \{0,1,\infty\}$, then $L_{inp}$ has the following exponent difference at $x=p$:
\begin{itemize}
	\item $\alpha_0 \cdot  e_p$ if $f$ has a zero at $x=p$ with multiplicity $e_p$,
	\item $\alpha_1  \cdot e_p$ if $f-1$ has a zero at $x=p$ with multiplicity $e_p$,
	\item $\alpha_{\infty}  \cdot e_p$ if $f$ has a pole at $x=p$ with order $e_p$.
\end{itemize}
\end{theorem}

\section{Computing Solutions of a Second Order Linear Differential Operator in the form of (\ref{form_of_solutions}) by using Quotients of Formal Solutions}
\label{section_algorithm}

This section gives our first algorithm, which looks for solutions of the form of  (\ref{form_of_solutions}).

\subsection{Problem Statement}\label{subsection_problem_statement}
Given a second order linear differential operator $L_{inp} \in \mathbb{Q}(x)[\partial]$, irreducible and regular singular, we want to find a $_2F_1$-type solution of the differential equation $L_{inp}(y)=0$ of the form of (\ref{form_of_solutions}). This is equivalent to finding transformations 1 and 2 from a GHDO $L_B$ to $L_{inp}$. Therefore, we need to find
\begin{enumerate}
	\item $L_{B}$ (i.e., find $a_1,a_2,b_1$),
	\item parameters $f$ and $r$ of the change of variables and exp-product transformations such that 
	$L_{B} \hspace{1mm} {\xrightarrow{f}}_C  \hspace{1mm} {\xrightarrow{r}}_E \hspace{1mm} L_{inp}$.
\end{enumerate}

\begin{alg}\label{find2f1}
General Outline of \texttt{find\_2f1}.

\begin{enumerate}

\item[] INPUT: $L_{inp} \in \mathbb{Q}(x)[\partial]$ and (optional) $a_fmax$ where
	\begin{enumerate}
	\item[] $L_{inp} =$ a second order regular singular irreducible operator,
	\item[] $a_fmax =$ bound for the algebraic degree $a_f$ (See Example \ref{example2}). If omitted, then $a_fmax=2$ which means our implementation tries $a_f=1$ and $a_f=2$.
	\end{enumerate}	
	
\item[] OUTPUT: 
Solutions of $L_{inp}$ in the form of (\ref{form_of_solutions}), or  an empty list.
\item[]
\item[] For each $a_f \in \{1, \ldots, a_fmax\}$:
\item 
Use Section \ref{subsection_cand_exp_diffs} 
to compute candidates for $L_B$ and $d_f$. 
This is the combinatorial part of the algorithm. 

\item 
For a candidate $(L_B, d_f)$, compute formal solutions of 
$L_B$ and $L_{inp}$ at a 
non-removable singularity (see Remark \ref{non-rem-sing} in Section 
\ref{subsection_transformations_and_singularities}) 
up to precision $a \geq 2(a_f+1)(d_f+1)+6$. 
Take the quotients of formal solutions and compute series 
expansions for $q^{-1}$ and $Q$ 
which will be used
to compute 
\begin{equation}\label{eq_c}
f = q^{-1}(c \, Q(x)) 
\end{equation}
in the next step.

\item 
Choose a good prime number $\ell$ 
and try to find $c$ mod $\ell$ by looping $c=1,2,\ldots,\ell-1$ as in Section 
\ref{subsection_quotient_method}.  For each $c$:

\begin{enumerate}

\item 
Compute $f$ mod $(x^a,\ell)$ from equation (\ref{eq_c}) and use it to reconstruct $f$ mod $\ell$ (the image of $f$ in $\mathbb{F}_\ell{}(x)$). 
If it fails for every $c$, then proceed with the next candidate GHDO 
(if any) in 
Step 2.
If no candidates remain, then
return an empty list.

\item 
If rational reconstruction in Step 3.1 succeeds for some $c$ values, then apply 
Hensel lifting (Section \ref{subsection_recovering_f_and_r}) 
to find $f$ mod a power of $\ell$. 
Then try rational number reconstruction. 
If it does not fail for at least one $c$ value, then we have $f$.
If no solution is found (see Remark \ref{address2} in Section~\ref{subsection_Lifting for a Rational Pullback Function}), then proceed with the next candidate GHDO 
(if any) in 
Step 2.
If no candidates remain, then
return an empty list.

\item 
Use Section \ref{subsection_transformations_and_singularities} to
compute the parameter $r$ of the exp-product transformation.

\end{enumerate}

\item 
Return a basis of $_2F_1$-type solutions of $L_{inp}$.
\end{enumerate}
\end{alg}

Step 2 is explained in Sections 
\ref{subsection_deg_bounds_for_pullback_functions} 
and \ref{subsection_cand_exp_diffs}. 
Step 3 is the quotient method, see Section 
\ref{subsection_quotient_method} for more. 
Steps 3.2 and 3.3 are explained in Sections 
\ref{subsection_recovering_f_and_r} and  \ref{r} respectively.
A Maple implementation of Algorithm \ref{find2f1} and some examples can be found at \cite{imamoglu1}.

\subsection{Degree Bounds for Pullback Functions}
\label{subsection_deg_bounds_for_pullback_functions}

\begin{theorem}[Riemann-Hurwitz Formula]\label{riemann-hurwitz-theorem}
Let $X$ and $Y$ be two algebraic curves with genera $g_X$ and $g_Y$ respectively. If $f : X \longrightarrow Y $ is a non-constant morphism, then
\begin{equation}
	\label{riemann-hurwitz-formula}
	2g_X - 2 = \deg(f) (2 g_Y - 2 ) + \sum_{p \in X} (e_p - 1).
\end{equation}
Here $e_p$ denotes the ramification order at $p \in X$.
See \cite{hartshorne} for more details.
\end{theorem}

Let $L_B \in \mathbb{Q}(x)[\partial]$ be a GHDO, $d_f := \deg(f)$, and assume that
\begin{equation*}
	L_B \hspace{1mm} 
	{\xrightarrow{f: \mathbb{P}^1 \longmapsto \mathbb{P}^1 }}_C  \hspace{1mm} 
	{\xrightarrow{r}}_E \hspace{1mm} 
	L_{inp}.
\end{equation*}  
Section 3.1 of \cite{imamoglu_vhoeij} gives 
an a priori bound for $d_f$,
\begin{equation}
\label{general_degree_bound}
	d_f \leq
	\begin{cases}
	6  (n_{true} - 2), & \text{logarithmic case,}\\
	36  \left(n_{true} - \frac{7}{3}\right), & \text{non-logarithmic case.}\\
	\end{cases}
\end{equation}
where $n_{true}$ is the number of true singularities of $L_{inp}$. Algorithm \ref{find2f1} uses this only as an initial degree bound.

\subsection{Riemann-Hurwitz Type Formula For Differential Equations}
\label{subsection_r_h_formula_for_des}

\begin{remark}\label{diffeqs_on_algcurves} 
Let $X$ be any algebraic curve and $\mathbb{C}(X)$ be its function field. The ring $D_{\mathbb{C}(X)} := \mathbb{C}(X)[\partial_t]$ is the ring of differential operators on $X$. Here $t \in \mathbb{C}(X) \setminus \mathbb{C}$. An element $L \in D_{\mathbb{C}(X)}$ is a differential operator defined on the algebraic curve $X$.
\end{remark}

\begin{theorem} \cite[Lemma~1.5]{Baldassari}
\label{rh_for_des_thm}
Let $X$, $Y$ be two algebraic curves with genera $g_X$, $g_Y$, and function fields $\mathbb{C}(X)$, $\mathbb{C}(Y)$ respectively. Let $f : X \longrightarrow Y$ be a non-constant morphism. The morphism $f$ corresponds to a homomorphism $\mathbb{C}(Y) \longrightarrow \mathbb{C}(X)$, which in turn corresponds to a homomorphism $D_{\mathbb{C}(Y)} \longrightarrow D_{\mathbb{C}(X)}$. If $L_1 \in D_{\mathbb{C}(Y)}$ with ${\rm ord}(L_1)=2$ and $L_2$ is the corresponding element in $D_{\mathbb{C}(X)}$, then
\begin{equation}
	\label{riemann-hurwitz-type-formula-thm}
		{\rm Covol}(L_2, X)  =   \deg(f) \cdot {\rm Covol}(L_1, Y)
\end{equation}
where 
\begin{equation*}
		{\rm Covol}(L, X) := 2 g_X - 2 + \sum_{p \in X} (1 -  \Delta(L,p))
\end{equation*}
and where $\Delta(L,p)$ is the absolute value of the exponent difference of $L$ at $p$.
\end{theorem}
\begin{proof}
Following \cite{Baldassari},
take finite sets $S \subseteq Y$ and $T = f^{-1}(S)$ in such a way that all singularities of $L_1$ are in $S$, all singularities of $L_2$ are in $T$, and all branching points in $X$ are in $T$ as well.
\begin{align}
		\#T  =  \sum_{p \in T} 1 & =  \sum_{p \in T}e_p + \sum_{p \in T} (1 - e_p) \\ 
	\label{numelemsT-1}
		& =  \deg(f) \cdot \#S + \sum_{p \in X} ( 1 - e_p ) \\ 
	\label{numelemsT-2}
	& = \deg(f) \cdot \#S - \left( 2g_X - 2 - \deg(f) ( 2g_Y - 2 ) \right).
\end{align}
From (\ref{numelemsT-1}) to (\ref{numelemsT-2}) we used (\ref{riemann-hurwitz-formula}). Then,
\begin{align}
	\sum_{p \in X} \left( 1 - \Delta(L_2,p)  \right) 
		& = \sum_{p \in T} \left( 1 - \Delta(L_2,p)  \right) \\ 
		& =  \sum_{p \in T} 1 \ - \ \sum_{p \in T} \Delta(L_2,p)  \\ 
		\label{RH-proof-sum-1}
		& = \#T \ - \ \deg(f) \sum_{s \in S} \Delta(L_1, s).
\end{align}
Then, combine (\ref{numelemsT-2}) and (\ref{RH-proof-sum-1}) and get
\begin{equation}
	\label{RH-proof-sum-2}
	2g_X - 2 + \sum_{p \in X} \left( 1 -  \Delta(L_2,p) \right)
	= \deg(f) \left(2g_Y - 2 + \sum_{s \in Y} \left(1 - \Delta(L_1,s) \right) \right)
\end{equation}
which is the same as (\ref{riemann-hurwitz-type-formula-thm}).
\end{proof}

\begin{corollary}
\label{Cor31}
Let $X=Y=\mathbb{P}^1$ and suppose that
$
L_B 
\hspace{1mm} {\xrightarrow{f:\mathbb{P}^1 \rightarrow \mathbb{P}^1 }}_C  
\hspace{1mm} {\xrightarrow{r}}_E \hspace{1mm} 
L_{inp}
$
where $L_B \in \mathbb{C}(x)[\partial]$ is a GHDO with exponent differences $[\alpha_0, \alpha_1, \alpha_{\infty}]$ at $\{0,1,\infty\}$. Since an exp-product transformation does not affect exponent differences, Theorem \ref{rh_for_des_thm} gives the following equation for  ${\rm Covol}(L_{inp},\mathbb{P}^1)$:
\begin{equation}
	\label{riemann-hurwitz-type-formula-for-P1}
		-2 + \sum_{p \in \mathbb{P}^1} (1 - \Delta(L_{inp},p)) 
	= \deg(f) \left( -2 +  \sum_{i \in \{0,1,\infty \}} (1 - \alpha_i ) \right).
\end{equation}
\end{corollary}

\begin{corollary}
Let $L_B$ and $L_{inp}$ be as in Corollary~\ref{Cor31}. 
Both have rational function coefficients.
This time, suppose that $f,r$
in
$
L_B 
\hspace{1mm} {\xrightarrow{f}}_C  
\hspace{1mm} {\xrightarrow{r}}_E \hspace{1mm} 
L_{inp}
$
are algebraic functions. Then $f : X \rightarrow \mathbb{P}^1$ for an algebraic curve $X$
whose function field $\mathbb{C}(X) = \mathbb{C}(x,f)$ is an algebraic extension of both $\mathbb{C}(x)\cong \mathbb{C}(\mathbb{P}^1)$ and $\mathbb{C}(f)  \cong \mathbb{C}(\mathbb{P}^1)$.
Let $a_f$ and $d_f$ denote the degrees of these extensions.
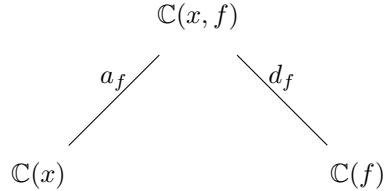
\begin{center}
\begin{tikzpicture}[-latex, auto, node distance=3cm, on grid, state/.style ={ circle }]
	\node[state] (C) {$\mathbb{C}(x,f)$};
	\node[state] (A) [below left=of C] {$\mathbb{C}(x)$};
	\node[state] (B) [below right=of C] {$\mathbb{C}(f)$};
	\path[-] (A) edge [bend right = 0] node[above] {$a_f$} (C);
	\path[-] (B) edge [bend left = 0] node[above] {$d_f$} (C);
\end{tikzpicture}
\end{center}
Applying (\ref{riemann-hurwitz-type-formula-thm}) to both field extensions gives:
\begin{equation}
	\label{rh_formula_af_is_not_1}
	{\rm Covol}(L_{inp}, \mathbb{P}^1) = \frac{d_f}{a_f}\left( -2 +  \sum_{i \in \{0,1,\infty \}} (1 - \alpha_i ) \right).
\end{equation}
\end{corollary}

\subsection{Candidate Exponent Differences}
\label{subsection_cand_exp_diffs}

This section explains how to obtain exponent differences for candidate GHDOs.

\vspace{11pt} 
\begin{alg}
General Outline of \texttt{find\_expdiffs}.
\begin{enumerate}

\item[] INPUT: $e_{inp}$, $e_{app}$, and $a_f$ where 
	\begin{enumerate}[label=(\roman*)]
	\item[] $e_{inp} =$ the list of exponent differences of $L_{inp}$ at 
	its true singularities,
	\item[] $e_{rem} =$ the (possibly empty) list of exponent differences 
	of $L_{inp}$ at its removable singularities,
	\item[] $a_f$ = candidate algebraic degree.
	\end{enumerate}

\item[] OUTPUT: A list of all lists 
$e_B = [\alpha_0, \alpha_1, \alpha_{\infty}, d]$ 
of integers or rational numbers where $[\alpha_0, \alpha_1, \alpha_{\infty}]$ 
is a list of candidate exponent differences and $d$ 
is a candidate degree $d_f$ for $f$ such that:

\begin{enumerate}[label=(\roman*)]
	\item For every exponent difference $m$ in $e_{inp}$ there 
	exists
	$e \in \mathbb{Q}$ with $e \cdot a_f \in \{1,\dots,d \}$ such that $m = e \cdot \alpha_{i}$ for some $i \in \{0,1,\infty \}$.
	 \item The multiplicities $e$ are consistent 
	with (\ref{riemann-hurwitz-formula}), and 
	their sums are compatible with $d$, see the last paragraph in Step $2$.
\end{enumerate}	
\item
Let $\overline{\alpha}_1, \overline{\alpha}_2, \overline{\alpha}_3$ 
$=$  $\alpha_0, \alpha_1, \alpha_{\infty}$. 
After reordering we may assume 
that $\overline{\alpha}_1$, $\dots$, $\overline{\alpha}_k$ $\in$ $\mathbb{Z}$ 
and $\overline{\alpha}_{k+1},\dots,\overline{\alpha}_3 \notin \mathbb{Z}$ 
for $k \in \{0,1,2,3\}$. 
For each $k \in \{0,1,2,3\}$ 
we use \texttt{CoverLogs} in \cite{imamoglu1} to compute candidates for 
$\overline{\alpha}_1,\dots,\overline{\alpha}_k \in \mathbb{Z}$.

Algorithm \texttt{CoverLogs} computes candidates that meet these requirements: 
\begin{itemize}
\item Logarithmic singularities are true singularities with integer exponent differences. If $L_{inp}$ has at least one logarithmic singularity $s$ with exponent difference $\Delta(L_{inp},s)$, then a candidate $L_B$ must have at least one logarithmic singularity; at least one of the $\overline{\alpha}_1, \overline{\alpha}_2, \overline{\alpha}_3$ must be an integer that divides $a_f \cdot \Delta(L_{inp},s)$,
and for every $\overline{\alpha}_i \in \mathbb{Z}$ there must be at least one $s$ such that $\overline{\alpha}_i$ divides $a_f \cdot \Delta(L_{inp},s)$.
\item $\Delta(L_{inp},s) = 0$ for some $s$ \ \ $\Longleftrightarrow$ \ \  $0 \in \{\overline{\alpha}_1, \overline{\alpha}_2, \overline{\alpha}_3\}$. 
\item Theorem \ref{theorem_expdiffs_and_multiplicity}. 
\end{itemize} 

If $\overline{\alpha}_1+\dots+\overline{\alpha}_k \neq 0$, then 
algorithm \texttt{CoverLogs} also computes the exact degree $d_f$ of $f$ using Theorem \ref{theorem_expdiffs_and_multiplicity} which shows that 
$d_f (\overline{\alpha}_1+\dots+\overline{\alpha}_k) / a_f$ must be the sum of 
the logarithmic exponent differences of $L_{inp}$. 
Otherwise, it uses (\ref{general_degree_bound}) to compute a bound for $d_f$, and uses it as  $d_f$ to compute a candidate degree.

\item
We will explain only the case $a_f=1$, and only $k=1$, 
which is the case 
$[\overline{\alpha}_1, \overline{\alpha}_2, \overline{\alpha}_3] 
= [\alpha_0, \alpha_1, \alpha_{\infty}]$,  
where $\alpha_0 \in \mathbb{Z}$ 
and $\alpha_1, \alpha_{\infty} \notin \mathbb{Z}$. 

\item[]
Let $k=1$. Let $ \alpha_0 \in \mathbb{Z}$ be one of the candidates from algorithm \texttt{CoverLogs}. 
We need to find candidates for $\alpha_1$ and $\alpha_{\infty}$. 

\item[]
The logarithmic singularities of $L_{inp}$ come from the point $0$. 
Non-integer exponent differences of $L_{inp}$ 
must be multiples of 
$\alpha_1$
or 
$\alpha_{\infty}$.  
Let $S_{N}$ be the set of non-logarithmic exponent 
differences of $L_{inp}$ and $S_{R}$ be the set of 
exponent differences of $L_{inp}$ at its removable singularities.  
Consider the set
\begin{equation*}
\Gamma_{1} = 
\begin{cases}
	\Gamma_{A}= \{ \frac{\max{(S_{N})}}{b}   :  
	b=1,\dots,d_f \} & \text{if $S_{N} \neq \emptyset$}, \\
	\Gamma_{B}= \{ \frac{a}{b}  :  a \in S_{R} \cup \{ 1 \}  ,   b=1,\dots,d_f  \} 
	& \text{otherwise}.
\end{cases}
\end{equation*}
$\alpha_1$ (or $\alpha_{\infty}$, but if so, we may interchange them)
must be one of the elements of $\Gamma_1$. 
We loop over all elements of $\Gamma_1$.
Assume that a candidate for $\alpha_1$ is chosen. 
Let $\Omega = S_N \setminus \alpha_1 \mathbb{Z}$. 
Now consider the set
\begin{equation*}
\Gamma_{\infty} = 
	\begin{cases}
		\Gamma_{A} \cup \Gamma_{B} & \text{if $\Omega = \emptyset$}, \\
		\{ \frac{g}{b} \, : \, g = \gcd{(\Omega)} : b=1,\dots,d_f  \} & \text{otherwise}.
	\end{cases}
\end{equation*}

\item[]
Now take all pairs $(\alpha_{\infty}, d)$ satisfying 
(\ref{rh_formula_af_is_not_1}), 
$\alpha_{\infty} \in \Gamma_{\infty}$, $1 \leq d \leq d_f$, 
with additional restrictions on $d$, as follows: 

\item[]
For every potential non-zero value $v$ for one of the $\alpha_i$'s 
we pre-compute a list of integers $N_v$ by dividing 
all exponent differences of $L_{inp}$ by $v$
and then selecting the quotients that are integers. 
Next, let $D_v$ be the set of all $1 \leq d \leq d_f$ 
that can be written as the sum of a sublist of $N_v$. 
Each time a non-zero value $v$ is taken for one of the $\alpha_i$, 
it imposes the restriction $d \in D_v$. 
This means that we need not run a loop for $\alpha_{\infty} \in \Gamma_{\infty}$,
instead, we run a (generally much shorter) loop for $d$ 
(taking values in the intersection of the $D_v$'s so far) 
and then for each such $d$ compute $\alpha_{\infty}$ from (\ref{rh_formula_af_is_not_1}).
We also check if $d \in D_{\alpha_{\infty}}$.

\item 
Return the list of candidate exponent differences with a candidate degree, 
the list of lists  $[\alpha_0, \alpha_1, \alpha_{\infty}, d]$, 
for candidate GHDOs.
\end{enumerate}
\end{alg}

\subsection{Quotient Method}
\label{subsection_quotient_method}
In this section, we explain a method to recover the pullback function $f$. We will explain our algorithm for rational pullback functions. For algebraic pullback functions, the only difference is the lifting algorithm, which is explained in Section \ref{subsection_recovering_f_and_r}. Note that we can always compute the formal solutions of a given differential equation $L_{inp}(y)=0$ up to a finite precision.

\subsubsection{Non-logarithmic Case}
\label{subsection_nonlog_quotient_method}

Let the second order differential equation $L_{inp}(y)=0$ be given. Let $L_{B}$ be a GHDO such that 
$
L_{B} \hspace{1mm} {
\xrightarrow{f}}_C \hspace{1mm}  
{\xrightarrow{r}}_E \hspace{1mm} 
L_{inp}.
$
Let $f : \mathbb{P}^1_x \mapsto \mathbb{P}^1_z$ and $L_1 \hspace{1mm} {\xrightarrow{f}}_C \hspace{1mm}L_2$. If $x=p$ is a singularity of $L_2$ and $z=s$ is a singularity of $L_1$, then we say that $p$ comes from $s$ when $f(p)=s$.

After a change of variables we can assume that $x=0$ is a singularity of $L_{inp}$ that comes from the singularity $z=0$ of $L_B$. This means $f(0)=0$ and we can write $f=c_0x^{v_0(f)}\left( 1 + \dots  \right)$ where $c_0 \in \mathbb{C}$, $v_0(f)$ is the multiplicity of $0$, and the dots refer to an element in $x \mathbb{C}[[x]]$.

Let $y_1$ and $y_2$ be the formal solutions of $L_B$ at $x=0$. The following diagram shows the effects of the change of variables and exp-product transformations on the formal solutions of $L_B$,
\begin{align*}
	y_i(x) 
	& \hspace{2mm} {\xrightarrow{f}}_C \hspace{2mm} y_i(f) \hspace{2mm} {\xrightarrow{r}}_E \hspace{2mm}  Y_i(x)=\exp{ (\int{r dx}) }y_i(f), \hspace{5mm}{i \in \{1,2\} }
\end{align*}
where $Y_1$ and $Y_2$ are solutions of $L_{inp}$. 

Let $q = \frac{y_1}{y_2}$ be a quotient of formal solutions of $L_B$. The change of variables transformation sends $x$ to $f$, and so $q$ to $q(f)$. Therefore, $q(f)$ will be a quotient of formal solutions of $L_{inp}$. 

The effect of exp-product transformation disappears under taking quotients. In general, a quotient of formal solutions of $L_B$ at a point $x=p$ is only unique up to M\"obius transformations $\frac{y_1}{y_2} \mapsto \frac{\alpha y_1+\beta y_2}{ \gamma y_1+ \eta y_2}$.

If $x=p$ has a non-integer exponent difference, then we can choose $q$ uniquely up to a constant factor $c$. So if we likewise compute a quotient $Q$ of formal solutions of $L_{inp}$, then we have $q(f)=c \cdot Q(x)$ for some unknown constant $c$. Then
\begin{equation}
	\label{formula1}
	f(x)= q^{-1}\left(c \cdot Q(x) \right).
\end{equation}

If we know the value of this constant $c$, then (\ref{formula1}) allows us to compute an expansion for the pullback function $f$ from expansions of $q$ and $Q$. To obtain $c$ with a finite computation, we take a prime number $\ell$. Then, for each $c \in \{1, \dots , \ell-1 \}$ we try to compute $f$ modulo $\ell$ in $\mathbb{F}_{\ell}(x)$ using series-to-rational function reconstruction. If this succeeds, then we lift $f$ modulo a power of $\ell$, and try to find $f \in \mathbb{Q}(x)$ with rational number reconstruction. Details of lifting are in Section \ref{subsection_recovering_f_and_r}.

\begin{remark}
\label{ansQ1} 
We compute formal solutions up to a precision $a \geq (a_f+1)(d_f+1)+6$. This suffices to recover the correct pullback function with a few extra terms to reduce the number of false positives.  
\end{remark}

\begin{alg}\label{case1}
General Outline of \texttt{case\_1: non\_logarithmic case}.
\begin{enumerate}
	
\item[] INPUT: $L_{inp}$, $L_B$, $d_f$, $a_f$, where
	\begin{enumerate}[label=(\roman*)]
	\item[] $L_{inp} =$ input differential operator,		
	\item[] $L_B =$ candidate GHDO,				
	\item[] $d_f =$ candidate degree for $f$,			
	\item[] $a_f =$ candidate algebraic degree for $f$.	
	\end{enumerate}
	
\item[] OUTPUT: $[f,r]$ or $0$, where
	\begin{enumerate}[label=(\roman*)]
	\item[] $f =$ pullback function,
	\item[] $r =$ parameter of exp-product transformation.
	\end{enumerate}
\item 
Compute formal solutions $y_1$, $y_2$ of $L_B$ 
and $Y_1$, $Y_2$ of $L_{inp}$ up to precision $a \geq (a_f+1)(d_f+1)+6$.

\item
Compute $q = \frac{y_2}{y_1}$, $Q = \frac{Y_2}{Y_1}$, and $q^{-1}$.

\item Select a prime $\ell$ for which these expansions can be reduced mod $\ell$.

\item
For each $c_0$ in $\{ 1,\dots, \ell-1 \}$: 
\begin{enumerate}
\item Evaluate $\overline{f}_{1,c_0} = q^{-1}(c_0 \cdot Q) \in \mathbb{Z}[x] / (\ell,x^{a})$.
\item If $a_f=1$ then try rational function reconstruction for $\overline{f}_{1,c_0}$   (the case $a_f > 1$ is explained in Section \ref{subsection_Lifting for  an Algebraic Pullback Function}).
\item[$\bullet$] If rational function reconstruction succeeds and produces $f_{1,c_0}$,
then store $c_0$ and $f_{1,c_0}$. 
\item[$\bullet$] If rational function reconstruction fails for every $c_0$, 
then return $0$.
\end{enumerate}

\item For $n$ from $2$ (see Remark \ref{address2} in Section~\ref{subsection_Lifting for a Rational Pullback Function}):
\begin{enumerate}
	\item[] For each stored $c_0$:
	
	\item
	Using the techniques explained in Section \ref{subsection_recovering_f_and_r} 
	lift
	$f_{n-1,c_0}$ to $f_{n,c_0}$.
	
	\item $f_{n,c_0}$ is a candidate for $f$ mod $\ell^n$. Try to obtain $f$ from this with rational number reconstruction.
	%
	%
	If this succeeds, compute $M$ such that $L_B \hspace{1mm} {\xrightarrow{f}}_C \hspace{1mm}M$. 
	Compute $r$ such that $M\hspace{1mm} {\xrightarrow{r}}_E \hspace{1mm} L_{inp}$, if it exists (see Section \ref{r}). If so, return $f$ and $r$. 
\end{enumerate}

\end{enumerate}
\end{alg}

\subsubsection{Logarithmic Case}
\label{subsection_logcase_quotient_method}

A logarithm may occur in one of the formal solutions of $L_{inp}$ at $x=p$ if exponents at $x=p$ differ by an integer. We may assume that $L_{inp}$ has a logarithmic solution at the singularity $x=0$.

Let $y_1$, $y_2$ be the formal solutions of $L_{B}$ at $x=0$. Let $y_1$ be the non-logarithmic solution (it is unique up to a multiplicative constant). Then $\frac{y_2}{y_1}=c_1 \cdot \log(x) + h$ for some $c_1 \in \mathbb{C}$ and $h \in \mathbb{C}[[x]]$. We can choose $y_2$ such that
\begin{equation}
	\label{c1}
	c_1=1 \,\,\, \text{ and } \,\,\,\,\, \text{constant term of $h$}=0.
\end{equation}
That makes $\frac{y_2}{y_1}$ unique. If $h$ does not contain negative powers of $x$ then define
\begin{equation}
	\label{g}
	g = \exp{\left(\frac{y_2}{y_1}\right)}=x\cdot \left( 1+\dots \right)
\end{equation}
where the dots refer to an element of $x\mathbb{C}[[x]]$.

\begin{remark}
\label{differently}
If we choose $y_2$ differently, then we obtain another $\tilde{g}= \exp{ \left( \frac{y_2}{y_1} \right) } $ that relates to $g$ in (\ref{g}) by $\tilde{g}=c_1g^{c_2}$ for some constants $c_1$, $c_2$. If $h$ contains negative powers of $x$, then the formula for $g$ is slightly different (we did not implement this case, instead we use Section \ref{section_integral_basis} to transform equations.).
\end{remark}

We do likewise for the formal solutions $Y_1$, $Y_2$ of $L_{inp}$ and denote 
\begin{equation}
	\label{G}
	G=\exp{\left(\frac{Y_2}{Y_1}\right)}=x\cdot \left( 1+\dots \right).
\end{equation}

Write $f \in \mathbb{C}(x)$ as $c_0 x^{v_0(f)} \cdot \left( 1+\dots \right)$. Then $g(f)=c\cdot x^{v_0(f)} \left( 1+\dots \right)$. Note that $g$, $G$ are not intrinsically unique, the choices we made in (\ref{c1}) implies that 
\begin{equation}
	\label{geq1}
	g(f)=c_1 \cdot G^{c_2}
\end{equation} 
for some constants $c_1$, $c_2$. Here $c_1 = c$ and $c_2 = v_0(f)$.

If $\Delta(L_{inp},0) \neq 0$, then find $v_0(f)$ from $\Delta(L_{B},0) v_0(f) = \Delta(L_{inp},0)$. Otherwise we loop  over $v_0(f)=1,2, \dots, d_f$. That leaves one unknown constant $c$. We address this problem as before, choose a good prime number $\ell$, try $c=1,2, \dots, \ell-1$. Then calculate an expansion for $f$ with the formula
\begin{equation}
	\label{geq2}
	f = g^{-1} \left( c \cdot G^{v_0(f)} \right).
\end{equation}

\begin{alg}
General Outline of \texttt{case\_2: logarithmic case}.
\begin{enumerate}

\item[] INPUT: $L_{inp}$, $L_B$, $d_f$, $a_f$,
where
	\begin{enumerate}[label=(\roman*)]
	\item[] $L_{inp} =$ input differential operator,
	\item[] $L_B =$ candidate GHDO,
	\item[] $d_f =$ candidate degree for $f$,
	\item[] $a_f =$ candidate algebraic degree for $f$.
	\end{enumerate}
\item[] OUTPUT: $[f,r]$ or $0$,
where
	\begin{enumerate}[label=(\roman*)]
	\item[] $f =$ pullback function,
	\item[] $r =$ parameter of exp-product transformation.
	\end{enumerate}
\item
Compute the exponents of $L_{inp}$ and $L_B$. 
If $\Delta(L_{inp},0)=0$, then replace $L_{inp}$ with $L$ 
defined in Remark (\ref{differently}) above. Otherwise let $L=L_{inp}$.

\item 
Compute formal solutions $y_1$, $y_2$ of $L_B$ 
and $Y_1$, $Y_2$ of $L$ up to precision $a \geq (a_f+1)(d_f+1)+6$.

\item 
Compute $q = \frac{y_2}{y_1}$, $Q = \frac{Y_2}{Y_1}$.
Compute $g$, $G$ from (\ref{geq1}) and (\ref{geq2}) respectively, and $g^{-1}$.

\item
Same as in Algorithm \ref{case1} Step 3.

\item
Compute $v_0(f)$ and search for $c_0$ value(s) such that  
$c$ could be $\equiv c_0$ mod $\ell$
by looping over $c_0=1,\dots,\ell-1$. 
If $\Delta(L_{inp},0)=0$, then also simultaneously 
loop over $v_0(f) = 1,\dots,d_f$ to find $v_0(f)$.

\item[] For each $c_0$ in $\{ 1,\dots,\ell-1 \}$:

\begin{enumerate}
\item Evaluate $\overline{f}_{1,c_0} = g^{-1} \left( c_0 \cdot G^{v_0(f)} \right) 
\in\mathbb{Z}[x] / (\ell,x^{a})$.

\item
Try rational function or algebraic function reconstruction for $\overline{f}_{1,c_0}$ as in Algorithm \ref{case1} Step 4.2.
%
%
\end{enumerate}

\item Same as in Algorithm \ref{case1} Step 5. 
	


\end{enumerate}
\end{alg}

\subsection{Lifting: Recovering the Pullback Function}
\label{subsection_recovering_f_and_r}

We explain lifting by using the formula (\ref{formula1}) for the pullback function, which occurs in the non-logarithmic case. The algorithm for the formula  (\ref{geq2}) in the logarithmic case is similar.

\subsubsection{Lifting for a Rational Pullback Function}
\label{subsection_Lifting for a Rational Pullback Function}
By using the formula (\ref{formula1}), which is $f(x)= q^{-1}\left(c \cdot Q(x) \right)$, we can recover the rational pullback function $f$, if we know the value of the constant $c$. We do not have a direct formula for $c$. However, if we know $c_0$ such that 
\begin{equation*}
	c \equiv c_0 \mod \ell
\end{equation*}
for a good prime number $\ell$, then we can recover the pullback function $f$. This can be done via \emph{Hensel lifting techniques}. 

Let $\ell$ be a good prime number and consider
\begin{align*}
	& h : \mathbb{Q} \longrightarrow \mathbb{Q}[x]/(x^{a}) \\
	& h(c) \equiv q^{-1}\left(c \cdot Q(x) \right) \mod x^{a}.
\end{align*}
By looping on $c_0 = 1, \dots, \ell-1$ and trying rational function reconstruction for $h(c_0)$ mod $(\ell,x^a)$, we can compute the image of $f \in \mathbb{F}_{\ell}(x)$ from its image in $\mathbb{F}_{\ell}[x] / ( x^a )$. If $a$ is high enough, then for correct value(s) of $c_0$, rational function reconstruction will succeed and return a rational function $\frac{A_0}{B_0}$ mod $\ell$. This $c_0$ is the one satisfying $c \equiv c_0$ mod $\ell$.

Write 
\begin{equation*}
c \equiv c_0 + \ell c_1 \mod \ell^2 
\end{equation*} 
for $0 \leq c_1 \leq \ell-1$. Taylor series expansion of $h$ gives us
\begin{equation}
	\label{equationh}
	h(c) = h(c_0 + \ell  c_1) \equiv h(c_0) + \ell   c_1  h'(c_0) \mod (\ell^2, x^a).
\end{equation}	
Substitute $c_1=0$, $c_1=1$, respectively, in (\ref{equationh}) and compute 
\begin{align}
	\label{a} & h(c_0) \mod (\ell^2, x^a), \\
	\label{b} & h(c_0+\ell) \equiv h(c_0)+\ell h'(c_0) \mod (\ell^2, x^a).
\end{align} 
Subtracting (\ref{a}) from (\ref{b}) gives
\begin{equation*}
	\ell h'(c_0) \equiv [ h(c_0+\ell) - h(c_0) ] \mod(\ell^2, x^a).
\end{equation*}
Let 
\begin{equation}
	\label{setS}
	E_{c_1} =  h(c_0) + c_1 \ell h'(c_0)
\end{equation} 
where $c_1$ is an unknown constant.
Suppose $f = \frac{A}{B}$ in characteristic $0$. We do not know what $A$ and $B$ are. However, from applying rational function reconstruction for $h(c_0)$, we obtain $A_0,B_0$ with $f \equiv \frac{A_0}{B_0}$ mod $(\ell ,x^a)$. It follows that
\begin{equation*}
f = \frac{A}{B} \equiv \frac{A_0}{B_0} \equiv E_{{c_1}}\mod(\ell, x^a).
\end{equation*} 
From this equation we have 
\begin{equation}
	\label{c}
	A \equiv B E_{{c_1}}\mod (\ell, x^a).
\end{equation}
Now let 
\begin{equation}
	\label{modp2}
	f = \frac{A}{B} \equiv \frac{A_0+\ell A_1}{B_0+ \ell B_1} \mod (\ell ^2,x^a)
\end{equation} 
where 
\begin{align*}
	A_1 & = a_0 + a_1x + \dots + a_{\deg(A_0)}x^{\deg(A_0)} \\
	B_1 & = b_1x + \dots + b_{\deg(B_0)}x^{\deg(B_0)} 
\end{align*}
are unknown polynomials. Here we are fixing the constant term of $B$. We need values of $\{a_i, b_j\}$ to find $f$ mod $(\ell^2,x^a)$. From (\ref{c}), we have 
\begin{equation}
	\label{solvethis}
	(A_0 + \ell  A_1) \equiv (B_0 + \ell  B_1) \cdot E_{c_1}  
	\mod (\ell^2, x^a).
\end{equation}
Now, solve the linear system (\ref{solvethis}) for unknowns $\{a_i, b_j, c_1\}$ in $\mathbb{F}_\ell$. From (\ref{modp2}) find $f$ mod $(\ell^2,x^a)$ and $c \equiv c_0 +\ell c_1$ mod $\ell ^2$. 

Try rational number reconstruction after each Hensel lift. If it succeeds, then check if this rational function is the one that we are looking for as in the last step of Algorithm \ref{case1}. If it is not, then lift $f$ mod $(\ell ^2,x^a)$ to mod $(\ell ^3,x^a)$ (or  $(\ell ^4,x^a)$ if an implementation for solving linear equations mod $\ell^n$ is available). After a (finite) number of steps, we can recover the rational pullback function $f$.

\begin{remark}\label{address2}
	Our implementation gives up when the prime power becomes ``too high''; (a proven bound is still lacking, but would be needed for a rigorous algorithm).
\end{remark}

\subsubsection{Lifting for  an Algebraic Pullback Function}
\label{subsection_Lifting for  an Algebraic Pullback Function}

We can recover algebraic pullback functions in a similar way. However, we need to know $a_f = [\mathbb{C}(x,f):\mathbb{C}(x)]$. The idea is to recover the minimal polynomial of $f$. 

Let $d_f = [\mathbb{C}(x,f):\mathbb{C}(f)]$.  Consider the polynomial in $y$
\begin{equation}
	\label{minpoly}
	\sum_{j=0}^{a_f}A_j  y^j \mod (\ell,x^a)
\end{equation}
with unknown polynomials
\begin{equation*}
	A_{j} = \sum_{i=0}^{d_f} a_{i,j} x^i
\end{equation*}
where $j=0,\dots,a_f$.

First we need to find the value of $c_0$ such that $c_0$ $\equiv$ $c$ mod $\ell$. As before, by looping on $c_0=1, \dots, \ell -1$, we can compute the corresponding $f_{c_0}$ which is a candidate for $f$ mod $(x^a, \ell)$ in $\mathbb{F}_{\ell}[x] / (x^a)$.
The polynomial (\ref{minpoly}) should be congruent to $0$ mod $(\ell,x^a)$ if we plug in $f_{c_0}$ for $y$. Solve the system
\begin{equation*}
	\sum_{j=0}^{a_f}A_j  f_{c_0} ^j \equiv 0 \mod (\ell ,x^a)
\end{equation*}
over $\mathbb{F}_\ell$ and find the unknown polynomials ${A_j}$ mod $\ell$. Then let 
\begin{equation*}
c \equiv c_0 + \ell c_1 \mod \ell^2. 
\end{equation*}
Now let $E_{c_0}$ be as in (\ref{setS}) and consider the system
\begin{equation*}
	\sum_{j=0}^{a_f}(A_j+\ell \tilde{A_j})  {E^j_{c_0}} \equiv 0 \mod (\ell^2,x^a).
\end{equation*}
Solve it over $\mathbb{F}_{\ell}$ to find $c_1$ and the unknown polynomials $\tilde{A_j}$. After a finite number of lifting steps and rational reconstruction, we will have the minimal polynomial of $\sum A_j y^j$ of $f$ in $\mathbb{Q}[x,y]$.

\subsection{Recovering the Parameter of Exp-product}\label{r}

After finding $f$, we can compute the differential operator $M$, such that 
\begin{equation*}
L_{B} 
\hspace{1mm} {\xrightarrow{f}}_C  \hspace{1mm} 
M 
\hspace{1mm} {\xrightarrow{r}}_E \hspace{1mm}
L_{inp}.
\end{equation*}
Then we can compare the second highest terms of $M$ and $L_{inp}$ to find the parameter $r$ of the exp-product transformation: If $M = \partial^2 + B_1 \partial + B_0$ and $L_{inp} = \partial^2 + A_1 \partial + A_0$, then 
\begin{equation*}
r =  \frac{B_1 - A_1}{2}.
\end{equation*}

\section{Computing an Integral Basis for a Linear Differential Operator}
\label{section_integral_basis}

We tested Algorithm \ref{find2f1} on many examples, including from the Online Encyclopedia of Integer Sequences, (\url{https://oeis.org}). Another source of examples comes from \cite{melou_mishna, bostan_kauers, bostan_chyzak_kauers_pech_vhoeij}
(see \cite{imamoglu3} for these operators).
Four of them have solutions in the form of (\ref{form_of_solutions}) and Algorithm \ref{find2f1} finds these solutions. However, Algorithm \ref{find2f1} does not solve the other operators from that list. We know from \cite{melou_mishna} and \cite{bostan_kauers} that these operators do have solutions in the form of (\ref{form_of_solutions_general_case}). It means that these operators must be gauge equivalent to operators with solutions in the form of (\ref{form_of_solutions}). The question is \emph{how can we find these gauge transformations}?
As mentioned in the introduction the key idea is to follow \cite{cohen} \texttt{POLRED}'s strategy; compute an integral basis and then normalize it at infinity. Then we can select an element with minimal valuations at infinity. This element gives us a gauge transformation. 

We modified the algorithm explained in \cite{kauers_koutschan}, which is an analogue of the algorithm in \cite{vhoeij_1994}, and implemented our own version of the integral basis procedure for second order regular singular linear differential operators. Our integral basis algorithm first finds local integral bases for each finite singularity of $L_{inp}$, then combines all of these local bases, and at the end normalizes the basis at infinity in the sense of \cite[Section 2.3]{trager} (see Section \ref{normalization}).

\subsection{Integral Bases}

\begin{definition}
	The \emph{local parameter} $t_p$ of a point $p \in \mathbb{C} \cup \{ \infty \}$ is defined as $t_p=x-p$ if $p \neq \infty$ and $t_p = \frac{1}{x}$ otherwise. If $L \in \mathbb{C}(x)[\partial]$ is regular singular, then all of the solutions of $L$ at a point $x=p$ are in the form
	\begin{equation*}
	f = t_p^{\nu_p}\sum_{i=0}^{\infty}P_i{t}_{p}^{i}
	\end{equation*} 
	where $\nu_p \in \mathbb{C}$ and $P_i \in \mathbb{C}[\log{(t_p)}]$ with $\deg{(P_i)} < {\rm ord}(L)$.
\end{definition}

\begin{definition} Let $y$ be a solution of $L$ at $x=0$,
	\begin{equation*}
	y = x^{\nu_0}\sum_{i=0}^{\infty}P_ix^{i}
	\end{equation*} where 
	$P_i \in \mathbb{C}[\log{(x)}]$.
	The \emph{valuation} of $y$ at $x=0$ is defined as follows:
	\begin{equation}
	\label{val_definition}
	v_0(y) := \nu_0 + \inf\{ i \, | \, P_i \neq 0 \}.
	\end{equation} 
	We say that $y$ is \emph{integral} at $x=0$ if 
	${\rm Re}( v_0(y) ) \geq0$.
\end{definition}

\begin{definition}
	If $p \in \mathbb{C} \cup \{ \infty \}$ and  $y$ is a solution of $L$ at $x=p$, then the valuation of $y$ at $x=p$ is defined as in (\ref{val_definition}) with $x$ replaced by $t_p$.
\end{definition}

\begin{definition}
	Let $L \in \mathbb{C}(x)[\partial]$ and $G \in \mathbb{C}(x)[\partial]$. The operator $G$ is called \emph{integral} for $L$ if 
	\begin{equation}
	\label{change}
	\forall_{p \in \mathbb{C}}{\rm Re}(v_p(G(y))) \geq 0.
	\end{equation}
	We may assume that ${\rm ord}(G) < {\rm ord}(L)$ because $G = QL + R$ for some $Q,R \in \mathbb{C}(x)[\partial]$ such that ${\rm ord}(R) < {\rm ord}(L)$, and we may replace $G$ by $R$ without changing $G(y)$ in (\ref{change}), i.e., we may interpret $G$ as an element of $\mathbb{C}(x)[\partial] / \mathbb{C}(x)[\partial]L$.
\end{definition}
\begin{definition}
	Let $L \in \mathbb{C}(x)[\partial]$ and
	\begin{equation*}
	M_L 
	=  
	\{ G \in \mathbb{C}(x)[\partial] \, | \, \text{$G$ is integral for $L$ } 
	\text{and } {\rm ord}(G) < {\rm ord}(L) \}.
	\end{equation*}
	A basis of $M_L$ as $\mathbb{C}[x]$-module is called an \emph{integral basis}.
\end{definition}

\subsection{Normalization at Infinity}
\label{normalization}
Assume that we computed an integral basis $[B_0, B_1]$ for a second order regular singular $L\in \mathbb{Q}(x)[\partial]$. One can normalize $[B_0, B_1]$ at infinity as follows:

Compute the formal solutions $Y_1$, $Y_2$ of $L_{inp}$ at $x=\infty$. Compute the pole orders (the pole order is the valuation as in (\ref{val_definition}) multiplied by $-1$) of 
\begin{equation}
\label{pole_orders}
B_0(Y_1), \hspace{5mm} B_0(Y_2),
\hspace{5mm}
B_1(Y_1), \hspace{5mm} B_1(Y_2)
\end{equation}
at the point $x=\infty$. Let the maximum of the pole orders be $m$ and let it come from $B_i(Y_j)$ where $i\in \{0,1\}$ and $j\in \{1,2\}$. Form the ansatz
\begin{equation*}
\mathfrak{B} = B_i(Y_j) - C \cdot x^{(m-n)} B_k(Y_j).
\end{equation*}
Here $k\in \{0,1\}$, $k \neq i$, and $n$ is the pole order of $B_k(Y_j)$ at $x=\infty$. For a suitable $C$, the pole order of $\mathfrak{B}$ at $x= \infty$ will be less than $m$. Find this $C$. Update $B_i = B_i - C \cdot  x^{(m-n)} B_k$. Now we have the updated basis $[ B_0, B_1 ]$. Compute the pole orders, for this updated basis, as in (\ref{pole_orders}). The possibilities for the updated basis $[ B_0,B_1 ]$ are
\begin{enumerate}
\item one of the pole orders decreases and none of them increases,
\item one of the other pole orders increases.
\end{enumerate}
If 1 occurs, then it means that we made an improvement, we are making the 
pole orders smaller. 
Repeat this process until there is no improvement possible (case 2). 
Then return the normalized basis $[ B_0,B_1 ]$.

\subsection{Finding a Suitable Gauge Transformation}

By using an integral basis $[ B_0,B_1 ]$, which is normalized at infinity, for $L_{inp}$, we want to find a gauge transformation $\mathcal{G}$ such that $\mathcal{G}$ transforms $L_{inp}$ to $\tilde{L}_{inp}$ such that 
$
L_{B} \hspace{1mm} 
{\xrightarrow{f}}_C \hspace{1mm}  
{\xrightarrow{r}}_E \hspace{1mm} 
\tilde{L}_{inp}.
$ 
We observed that for the operators coming from \cite{melou_mishna} and \cite{bostan_kauers}, one of the basis elements always gives such a gauge transformation $\mathcal{G}$. We tested our main algorithm on other examples as well and it turns out that this approach is very effective. 

\begin{alg} General Outline of \texttt{hypergeometricsols}.\label{hypergeometricsols}
	\begin{enumerate}
		\item[] INPUT: $L_{inp} \in \mathbb{Q}(x)[\partial]$ and (optional) $a_fmax$ where
			\begin{enumerate}
				\item[] $L_{inp} =$ a second order regular singular irreducible operator,
				\item[] $a_fmax =$ bound for the algebraic degree $a_f$. If omitted, then $a_fmax=2$ which means our implementation tries $a_f=1$ and $a_f=2$.
			\end{enumerate}	
		
		\item[]
		OUTPUT: Solutions of $L_{inp}$ in the form of (\ref{form_of_solutions_general_case}),  
		or  an empty list.
		
		\item Try to find solutions of $L_{inp}$ in the form of (\ref{form_of_solutions}) by using the Algorithm \ref{find2f1} in Section \ref{section_algorithm}. If none are found go to Step 2. 
		
		\item Compute an integral basis $[B_0, B_1]$ 
		for 
		$L_{inp}$ and normalize this basis at infinity by using the method
		given in Section \ref{normalization}. Each basis element $B_k$, $(k \in \{0,1\})$, 
		is a candidate gauge transformation. Transform $L_{inp}$ to another operator using $B_k$. Try to find solutions of the new operator in the form of (\ref{form_of_solutions}) by using Algorithm \ref{find2f1}. If 
		this new operator has solutions of type (\ref{form_of_solutions}), then apply the inverse of the gauge transformation to these solutions to form the solutions of $L_{inp}$ of type (\ref{form_of_solutions_general_case}), and return them. Otherwise return an empty list.
	\end{enumerate}
\end{alg}
\noindent A Maple implementation of Algorithm \ref{hypergeometricsols} and examples can be found at \cite{imamoglu3}.

\bibliographystyle{elsart-harv}
\bibliography{refs}


\end{document}